\documentclass[aps,final,twocolumn]{revtex4-1}
\usepackage{amsmath,amssymb,amsfonts,amsthm,color,graphicx}

\newcommand{\commentout}[1]{}
\makeatother

\definecolor{nblack}{rgb}{0,0,0}
\definecolor{nblue}{rgb}{0.2,0.2,0.7}
\definecolor{nred}{rgb}{0.7,0.2,0}
\definecolor{ngreen}{rgb}{0.0,0.8,0.2}

\newtheorem*{definition}{Definition}

\newtheorem{theorem}{Theorem}
\newtheorem*{observation}{Observation}

\begin{document}

\title{No-go theorems for $\psi$-epistemic models based on a continuity assumption}

\author{M. K. Patra$^{1}$, S. Pironio$^{1}$, and S. Massar$^{1}$}
\affiliation{$^{1}$Laboratoire d'Information Quantique, Universit{\'e} libre de Bruxelles (ULB), 1050 Bruxelles, Belgium}

\begin{abstract}
The quantum state $\psi$ is a mathematical object used to determine the probabilities of different outcomes when measuring a physical system. Its fundamental nature has been the subject of discussions since the inception of quantum theory: is it \emph{ontic}, that is, does it correspond to a real property of the physical system? Or is it \emph{epistemic}, that is, does it merely represent our knowledge about the system?
Assuming a natural continuity assumption and a weak separability assumption, we show here that epistemic interpretations of the quantum state are in contradiction with quantum theory. Our argument is different from the recent proof of Pusey, Barrett, and Rudolph and it already yields a non-trivial constraint on $\psi$-epistemic models using a single copy of the system in question. 
\end{abstract}
\maketitle

\emph{Introduction.}
Quantum theory textbooks usually start from the hypothesis that to every physical system corresponds a mathematical object -- a ray in Hilbert space -- called the {\em quantum state}. They then go on to deduce the multitude of quantitative predictions that make quantum theory so successful. But does the quantum state correspond to a real physical state or does it merely represent an observer's knowledge about the underlying reality? A major reason for doubting the reality of the quantum state is that it cannot be observed directly: it can only be reconstructed indirectly by lengthy state estimation procedures \cite{Estimation,WeakMeas}. Furthermore, an epistemic interpretation of the quantum state could provide an intuitive explanation for many counterintuitive quantum phenomena and paradoxes, such as the measurement postulate and wavefunction collapse
 \cite{CFS2002,bz,Spekkens2007}.

To formulate with precision the above question, we assume, following \cite{HarriganSpekkens}, that every quantum system possesses a real physical state (also called ontic state), denoted $\lambda$, which is independent of the observer. When a measurement is performed on the system, the probabilities to get different outcomes are determined by  $\lambda$. If an ensemble of such systems is prepared, different members of the ensemble may be found in different states $\lambda$. A preparation procedure $Q$ therefore corresponds in general to a probability distribution $P(\lambda|Q)$ over the real states. The probability to obtain the outcome $r$ when preparation $Q$ is followed by measurement $M$ is $P(r|M,Q) = \sum_{\lambda}  P(r|M,\lambda)P(\lambda|Q)$. Such a model will reproduce the quantum predictions if $P(r|M,Q)=\langle \psi_Q |\mathcal{M}_r|\psi_Q\rangle$, where $\psi_Q$ is the quantum state assigned by quantum theory to the preparation $Q$ and $\mathcal{M}_r$ is the quantum operator describing the measurement.

We can now distinguish two classes of models of the above type. A model is said to be $\psi$\emph{-ontic} if the preparation of distinct pure quantum states always give rise to distinct real states. That is, for every $\lambda$ either $P(\lambda|Q)=0$ or $P(\lambda|Q')=0$ if the preparations $Q$ and $Q'$ correspond to different quantum states $|\psi_Q\rangle \neq |\psi_{Q'}\rangle$. In this case, every real state $\lambda$ is compatible with a unique pure quantum state. The quantum state is ``encoded'' in $\lambda$ and we can consider it to represent a real property of the system, akin, e.g., to the total energy of a system in classical physics \cite{Leifer11}. In the second class of models, known as $\psi$\emph{-epistemic} models, preparation of distinct pure quantum states may result in the same real state $\lambda$. Formally, there exists preparations $Q$ and $Q'$ corresponding to distinct quantum states  $|\psi_Q\rangle \neq |\psi_{Q'}\rangle$ such that both $P(\lambda|Q)>0$ and $P(\lambda|Q')>0$ for some $\lambda$. In this case, the quantum state is not uniquely determined by the underlying real state and has a status analogous, e.g., to the Liouville distribution in statistical physics.

While non-trivial $\psi$-epistemic models exist in any fixed dimension $d$ \cite{QSCanBe,ABCL13}, such models are necessarily highly contrived. Indeed, Pusey, Barrett, and Rudolph (PBR) have recently shown that the predictions of $\psi$-epistemic models are in contradiction with quantum theory under the assumption, termed \emph{preparation independence}, that independently prepared pure quantum states correspond to product distributions over ontic states \cite{QSCannotBe}. In the present work, we derive two alternative no-go theorems for $\psi$-epistemic models based on a natural assumption of continuity. Our approach shows that already at the level of a single system there exist strong constraints on $\psi$-epistemic models. Furthermore, our first no-go theorem readily translates in a simple experimental test, an implementation of which has been reported in \cite{arxiv1} using high-dimensional attenuated coherent states of light traveling in an optical fibre.

Constraints on $\psi$-epistemic models at the level of single quantum systems have also been obtained in \cite{hardy12} using an assumption termed \emph{ontic indifference}. This assumption is in fact closely related to the one presented here.
In Section 2 of the Supplemental Material~\cite{SuppMat}, we show how to use our approach to recover, in a simple and clear way, those of \cite{hardy12}. Arguments using single quantum systems have also been used in \cite{M12,LM12} to show that $\psi$-epistemic model cannot be ``maximally epistemic" (in a sense defined in \cite{M12,LM12}.

\emph{No-go theorems for $\psi$-epistemic models.}
The key motivation behind our result is that  $\psi$-epistemic models should satisfy a form of continuity. 
Indeed, we assign an ontic status to $\psi$ if a variation of $\psi$ necessarily implies a variation of the underlying reality $\lambda$, and we assign it an epistemic status if a variation of $\psi$ does not necessarily imply a variation of $\lambda$. It is then natural to assume a form of continuity for $\psi$-epistemic models:  a slight change of $\psi$  induces a slight change  in the corresponding ensemble of $\lambda$'s in such a way that at least some $\lambda$'s from the initial ensemble will also belong to the perturbed ensemble. We use a slightly stronger form of continuity which asserts that there are real states $\lambda$ in the initial ensemble that will remain part of the perturbed ensemble, no matter how we perturb the initial state, provided this perturbation is small enough. Models that violate this condition are presumably very contrived. 
Formally this continuity condition is defined as follows (see Fig.~\ref{fig:Model} for a depiction of the difference between $\psi$-ontic and $\delta$-continuous $\psi$-epistemic models).

\begin{figure} \begin{center}
\includegraphics[scale=0.4]{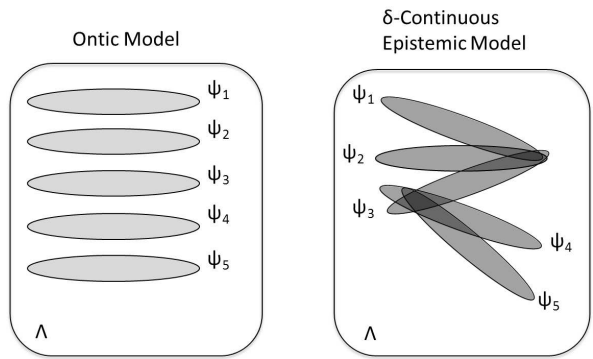}
\caption{Illustration of $\psi$-ontic and $\delta$-continuous $\psi$-epistemic models. Depicted is the space $\Lambda$ of ontic states, as well as the support of the probability distribution $P(\lambda|Q_k)$ for preparation $Q_k$ associated to distinct pure states $\psi_k$, $k=1,\ldots,5$. In $\psi$-ontic models (left) distinct quantum states give rise to probability distribution $P(\lambda|Q_k)$ with no overlap. In $\delta$-continuous $\psi$-epistemic models (right), states that are close to each other (such as $\{\psi_1,\psi_2,\psi_3\}$ and $\{\psi_3, \psi_4,\psi_5\}$) all share common ontic states. However states that are further from each other (such as $\psi_1$ and $\{\psi_4,\psi_5\}$) do not necessarily have common ontic states $\lambda$.}
\label{fig:Model}
\end{center} \end{figure}

\begin{definition}[$\delta$-continuity]
Let $\delta>0$ and let $B_\psi^\delta$ be the ball of radius $\delta$ centered on $|\psi\rangle$, i.e., $B_\psi^\delta$ is the set of states $|\phi\rangle$ such that $|\langle\phi|\psi\rangle| \geq 1-\delta$. We say that a model is $\delta$-continuous if for any preparation $Q$, there exists an ontic state $\lambda$ (which can depend on $Q$) such that for all preparations $Q'$ corresponding to quantum states $|\phi_{Q'}\rangle$ in the ball $B^\delta_{\psi_Q}$ centered on the state $|\psi_Q\rangle$, we have $P(\lambda|Q')>0$. 
\end{definition}

Note that for notational simplicity we  formulate our results in the case where the set $\Lambda=\{\lambda\}$ of real states is finite or denumerable. The generalisation of Theorems~1 and 2 below to measurable spaces is given in Section 1 of the Supplemental Material~\cite{SuppMat}. This generalisation is important since reproducing the predictions of even a single qubit requires an infinite, and probably even uncountably infinite, number of real states, see \cite{hardy04,Montina08,Montina11} for evidence to this effect.

Note also that the above definition introduces a connection between the overlap of quantum states and the overlap of distributions in the ontic space of $\lambda$s. This is extremely natural if we do not introduce a privileged direction in Hilbert space (i,e, a preferred basis), since then the properties of psi-epistemic models can only depend on the geometry of the Hilbert space.

Our first result is a constraint on $\delta$-continuous models for single systems. 
\begin{theorem} There are no $\delta$-continuous models with $\delta \geq 1-\sqrt{(d-1)/d}$ reproducing the measurement statistics of quantum states in a Hilbert space of dimension $d$.
\end{theorem}
\begin{proof}
Consider $d$ preparations $Q_k$ ($k=1,\ldots,d$) corresponding to distinct quantum states $|\psi_k\rangle$ all contained in a ball of radius $\delta$. By definition of a $\delta$-continuous model, there is at least one $\lambda$ for which $\min_{k}P(\lambda|Q_k) > 0$ and thus 
\begin{equation} \epsilon \equiv \sum_\lambda \min_{k}P(\lambda|Q_k) > 0\,.
\end{equation}
This last quantity can be viewed as a measure of the extent to which distributions over real states overlap in the neighborhood of a given quantum state. It was also introduced in \cite{QSCannotBe} where it was shown to be related to the variational distance between the distributions $P(\lambda|Q_k)$. 

Suppose now that a measurement $M$ yielding one of the possible outcomes $r=1,\ldots,d$ is made on each of the prepared systems. A $\delta$-continuous model then makes the prediction
\begin{eqnarray} 
\sum_{k}P(k|M, Q_k) &=& \sum_{k}\sum_{\lambda}P(k|M,\lambda)P(\lambda|Q_k)\nonumber\\
& \geq &\sum_{k}\sum_{\lambda}P(k|M,\lambda)\min_{k}P(\lambda|Q_k)\nonumber\\
& =& \sum_{\lambda}\min_{k}P(\lambda|Q_k) = \epsilon > 0 \, .\label{eq:NoGo}
\end{eqnarray}
According to quantum theory, however, there exist states in a Hilbert space of dimension $d$ contained in a ball of  radius $\delta = 1-\sqrt{(d-1)/d}$ such that the left-hand side of Eq.~(\ref{eq:NoGo}) is equal to $0$. To show this, let $\{|j\rangle \,:\, j=1,\ldots,d\}$ be a basis of the Hilbert space. Consider the $d$ distinct states $|\psi_k\rangle = \frac{1}{\sqrt{d-1}}\sum_{j\neq k}|j\rangle$. These states are all at mutual distance $|\langle\psi_k|\psi\rangle| = \sqrt{(d-1)/d}$ from the state $|\psi\rangle = \frac{1}{\sqrt{d}}\sum_{j}|j\rangle$. Let the measurement $M$ be the measurement in the basis $\{|j\rangle\}$. Then $P(k|M,Q_k)=0$ for all $k=1,\ldots,d$ and thus $\sum_{k}P(k|M,Q_k)=0$. (These states and measurements were considered in the $d=3$ case in \cite{CFS2002b}).
\end{proof}
Note that the above result also applies if we only require $\delta$-continuity to hold around some fixed quantum states rather than for all states in Hilbert space. Interestingly, the $\psi$-epistemic model of \cite{QSCanBe} for Hilbert spaces of dimension $d$  is $\delta$-continuous around a specific state with a value of $\delta$ saturating the above bound.

Though one expects a $\psi$-epistemic model to be $\delta$-continuous for some value of $\delta$, the bound derived in Theorem~1 may a priori seem arbitrary. This motivates the following definition.

\begin{definition}[Continuity] A $\psi$-epistemic model is continuous if there exists a non-zero $\delta>0$ such that it is $\delta$-continuous. 
\end{definition}

Our second result shows that
there are no $\psi$-epistemic models that are both continuous and which satisfy the following separability assumption. A similar condition was independently introduced in \cite{SchlosshauerFine12}, where it is called ``compactness". Though weaker than the preparation independence assumption explicitly used by PBR, it is already sufficient to derive their main result. 

\begin{definition}[Separability]
Let $Q$ be the preparation of a physical system yielding with non-zero probability $P(\lambda|Q)>0$ the real state $\lambda$. A model is separable if  $n$ independent copies $Q^n=(Q,\ldots,Q)$ of the preparation devices yield with non-zero probability $P(\vec{\lambda}=\lambda^n|Q^n)>0$ a system in the joint real state $\lambda^n=(\lambda,\ldots,\lambda)$, for any positive integer $n$.
\end{definition}
\begin{theorem}
Separable continuous $\psi$-epistemic models cannot reproduce the measurement statistics of quantum states in a Hilbert space of dimension $d\geq 3$.
\end{theorem}
\begin{proof}
The idea of the proof is to fix an arbitrarily small $\delta>0$, and choose specific states $|\phi_{k}\rangle$, all within the ball of radius $\delta$. Because of $\delta$-continuity, these states share a common ontic state. Using separability, the states $|\phi_k^{\otimes n}\rangle$ also share a common ontic state. By taking $n$ large enough, the distance between the tensor products $|\phi_k^{\otimes n}\rangle$ becomes large enough that we can apply Theorem 1. More in detail we proceed as follows.

Consider $d\geq 3$ preparations $Q_k$ corresponding to the $d$ distinct states $|\phi_{k}\rangle=\alpha|k\rangle+{\beta}/{\sqrt{d}}\sum_{i=1}^d|i\rangle$ with $k=1,\ldots,d$, $\alpha=-\sqrt{1-[(d-2)/(d-1)]^{1/n}}$ and $\beta=-{\alpha}/{\sqrt{d}}+\sqrt{{\alpha^2}/{d}+[(d-2)/(d-1)]^{1/n}}$.
It is easily checked that these states are normalised $ \langle \phi_k\vert\phi_k\rangle=1$, have mutual scalar product $|\langle\phi_k|\phi_{l}\rangle|=[(d-2)/(d-1)]^{1/n}$ for $k\neq l$, and are all at distance $|\langle\phi_k|\phi\rangle|=1-\delta_{nd}$  from the state $|\phi\rangle=\frac{1}{\sqrt{d}}\sum_{i=1}^d |i\rangle$, where $\delta_{nd}=1-\sqrt{1-(d-1)\alpha^2/d}$. In a $\delta_{nd}$-continuous model, these states share at least a common real state $\lambda$. The separability assumption then implies that the states $|\phi_k^{\otimes n}\rangle$ also share a common real state, and thus that $\epsilon_n\equiv \sum_\lambda \min_{k}P(\lambda|Q^n_k) > 0$ . By the same argument as in Theorem~1, it then follows that if a measurement $M$ yielding one of the possible outcomes $r=1,\ldots,d$ is performed on each of these systems, the quantity $\sum_k P(k|M,Q^n_k)\geq \epsilon_n>0$. 

Note now that the $d$ $n$-fold copies $|\phi_k^{\otimes n}\rangle$ are normalised and have mutual scalar product $|\langle\phi_k^{\otimes n}|\phi_l^{\otimes n}\rangle|=(d-2)/(d-1)$ for $k\neq l$. There therefore exists a unitary transformation $U$ in the subspace $S_d\subset \mathbb{C}_d^{\otimes n}$ spanned by the $d$ states $|\phi_k^{\otimes n}\rangle$ that carries out the transformation: $U|\phi_k^{\otimes n}\rangle=|\psi_k\rangle$ where $|\psi_k\rangle=\frac{1}{\sqrt{d-1}}(\sum_{j=1}^d|j\rangle-|k\rangle)$ for some  basis $\{|j\rangle\}$. The states $|\psi_k\rangle$ are identical to the states used in the proof of Theorem~1. It follows that there exists a $d$-outcome measurement $M$ in $\mathbb{C}_d^{\otimes n}$ which applied on the states $|\phi_k^{\otimes n}\rangle$ gives the same statistics as the measurement in the basis $\{|j\rangle\}$ applied on the states $|\psi_k\rangle$ of Theorem 1.  We can therefore find a measurement such  that $P(k|M,Q_k^n)=0$ and thus $\sum_k P(k|M,Q_k^n)=0$ in contradiction with the prediction of a $\delta_{nd}$-continuous separable $\psi$-epistemic model.

We have thus shown that one can exclude $\delta$-continuous separable models with $\delta \geq \delta_{nd}$ for any positive integers $d$ and $n$. 
For large $n$ this bound behaves as $\delta \gtrsim {\gamma}/{n}$, with $\gamma = (d-1)[\log(d-1)-\log(d-2)]/(2d)$, thereby implying by taking $n$ arbitrarily large that no $\psi$-epistemic model for Hilbert spaces of dimension $d \geq 3$ can satisfy both the assumptions of {continuity}  and {separability}. 
\end{proof}

It is interesting to compare how these no-go theorems could be used in practice to test $\psi$-epistemic models (see \cite{arxiv1,Nigg12} for actual tests). Experimental tests that rule out $\psi$-epistemic models for smaller values of the continuity parameter $\delta$ are clearly stronger. We thus consider how resources scales as $\delta \rightarrow 0$. If we use the construction of Theorem 1, then we need to use systems of dimension $d=O(1/(2\delta))$, and the resources needed to test $\psi$-epistemic models increase as $O(1/\delta)$. If we use Theorem 2, and take e.g. the dimension $d=3$, then one needs to prepare three states $|\phi_1^{\otimes n}\rangle,|\phi_2^{\otimes n}\rangle,|\phi_3^{\otimes n}\rangle$, where the number of copies of each state is
$n=O(\ln 2/(3\delta))$. Again the resources needed  increase as $O(1/\delta)$. Finally, we could also test $\delta$-continuous $\psi$-epistemic models using the construction given in PBR \cite{QSCanBe}. In this case we need to prepare $2^n$ distinct states, each of which is a product  state of $n$ qubits, with $n=O(\sqrt{2}\ln 2/\sqrt{\delta})$. The resources required for the application of the PBR construction therefore grow exponentially in $1/\sqrt{\delta}$. Experimental tests based on Theorems 1 and 2 thus seem much easier than those based on the PBR construction.

Note that our theoretical arguments and PBR's one rely on the fact that certain quantum probabilities are exactly equal to zero. However, these arguments are robust against small deviations from these predictions, as expected in an experimental implementation where noise is inevitably present.
Indeed, Eq.~(\ref{eq:NoGo}) implies that the observed value $\epsilon_\mathrm{exp}=\sum_k P(k|M,Q_k)$ provides  an upper bound on the overlap
$\epsilon =\sum_\lambda \min_{k}P(\lambda|Q_k)$ of the ontic distributions $P(\lambda|Q_k)$.
A small value of $\epsilon_\mathrm{exp}$ therefore translates into a strong constraint on continuous $\psi$-epistemic models, since it implies that these distributions have only a small common overlap.
Similarly $\epsilon^n_\mathrm{exp}=\sum_k P(k|M,Q^n_k)$ in Theorem~2 upper bounds the overlap $\epsilon_n =\sum_{\vec{\lambda}} \min_{k}P(\vec{\lambda}|Q^n_k)$ of the $n$-copy joint distributions $P(\vec{\lambda}|Q^n_k)$. This last quantity can simply be related to the single-copy overlap $\epsilon$ if we further make the preparation independence assumption of PBR that joint distributions $P(\vec{\lambda}|Q^n_k)=P(\lambda_1,\ldots\lambda_n|Q^n_k)=P(\lambda_1|Q_k)\ldots P(\lambda_n|Q_k)$ are product of individual distributions, which then implies $\epsilon_n=\epsilon^n$.
Note that a comparison of the sensitivity of the different tests to experimental noise is possible, but goes beyond the scope of the present work. It would require a detailed modeling of the state preparation and measurement procedures.

\emph{Discussion.}
In his seminal paper on the probabilistic interpretation of quantum theory, Born gave the wavefunction a functional interpretation: a mathematical object from which the probabilities of different measurement outcomes can be determined \cite{Born}. But the fundamental nature of this object, a real physical wave or a summary of our knowledge about physical systems, is a question that has divided physicists ever since. A precise formulation of these two alternatives, opening the way to clear-cut answers, was provided by Harrigan and Spekkens \cite{HarriganSpekkens}. 
If the wavefunction corresponds to a real, ontic, property of physical systems, the preparation of a system in different pure quantum states should always result in different physical states. If, on the other hand, the wavefunction has an epistemic status, such preparations should sometimes result in the same underlying physical state. PBR have recently introduced a no-go theorem that, given certain assumptions, rules out this latter possibility \cite{QSCannotBe}, thus awarding ontic status to the wavefunction. This result can also be seen as a constraint on the structure of possible extensions or generalizations of quantum theory. If they reproduce the quantum predictions and satisfy these assumptions, then such theories can only \emph{supplement} the wavefunction $\psi$ with additional variables $\lambda'$, i.e., a system should be described by a physical state of the form $\lambda=(\psi,\lambda')$.

The theorem of PBR relies on two main assumptions. The first, which is also unquestioned in the present work, is that a system has a real and objective state $\lambda$ that is independent of the observer. The second is an assumption of ``preparation independence", which states that independently prepared systems are described by independent product distributions over real states. It can be replaced by the  weaker separability assumption used here. 

In the present work, we reached the same conclusion as PBR using a simple argument that relies on a natural assumption of continuity. This notion of continuity captures the intuition that in a model where the quantum state is epistemic, a small variation of $\psi$ does not necessarily imply a variation of the underlying real state $\lambda$. 
We derived a fundamental limit on the degree of continuity of $\psi$-epistemic models, as parametrized by a quantity $\delta$, already at the level of single quantum systems (Theorem~1). 
Combining our continuity assumption with a separability assumption, we then showed that no $\psi$-epistemic model can reproduce all the predictions of quantum theory (Theorem~2). 

Besides their simplicity and the fact that they already constrain $\psi$-epistemic models for single quantum systems, an interest of our results is that they are easy to implement experimentally. Such an experimental test based on Theorem~1 has been reported in \cite{arxiv1}. 

\emph{Acknowledgments.}
M.~P. would like to thank Rob Spekkens for helpful discussions. We thank anonymous referees for useful comments. We acknowledge financial support from the European Union under project QCS, from the FRS-FNRS under project DIQIP, and from the Brussels-Capital Region through a BB2B grant.

\appendix
\section{General formulation of the no-go theorems in the context of measurable spaces}\label{appA}

Here we formulate our main result in the context of measurable spaces. Two main cases are of interest: when the space of ontic states $\Lambda$ is discrete, and when it is absolutely continuous. The discrete case was treated in the main text. It has the advantage of notational simplicity, but is probably difficult to justify as it is known that reproducing the predictions of quantum theory requires that the space of ontic states $\Lambda$ be continuous \cite{hardy04,Montina08,Montina11}. The absolutely continuous case is more natural, but the notation and proofs are slightly more cumbersome. In the following we first introduce a  general measure theoretic framework for ontic states, and then restrict to the absolutely continuous case. 

The state space $\Lambda$ is assumed to be a (probability) measure space with a fixed $\sigma$-algebra $\Omega$ of events. We assume the existence of a reference measure $m$ on $\Lambda$. For a measurable set $A \subset \Lambda$, $m(A)=\int \chi_A {\mathrm dm}(\lambda)$, where $\chi_A$ is the indicator function for $A$ (1 on $A$ and 0 elsewhere).

Each preparation $Q$ induces a conditional probability measure $P(-|Q)$ defined on $\Omega$. From the Lebesgue decomposition theorem \cite{Rudin_RealAnalysis} we infer that the measure $P(-|Q)$ can be decomposed as
\[ P(-|Q) = P_s(-|Q) + P_a(-|Q) \, , \] 
where the measures $P_s(-|Q)$ and $P_a(-|Q)$ are respectively singular and absolutely continuous with respect to $m$. Recall that a measure $\mu$ is concentrated on a set $C$ if for any measurable set $B$, $\mu(B)=\mu(B\cap C)$, and that two measures are singular with respect to each other if they are concentrated on disjoint sets. Furthermore the singular part $P_s(-|Q)$ can be decomposed into a discrete singular part which consists of a sum of Dirac measures (``$\delta$ functions'') and a singular continuous part.

In the remainder we consider the case where $P(-|Q)$ only consists of an absolutely continuous part. This is the assumption that is generally made in the literature. The case where there is only a discrete singular part was discussed in the main text. Note that our results can be extended to the case where either the absolutely continuous or the discrete singular part are non-zero. Our results do not apply to the case where there is only a singular continuous part.

We denote by $S_m$ the set on which $m$ is concentrated. The fact that the measure $P(-|Q)$ is absolutely continuous with respect to $m$ ($P(A|Q)=0$ whenever $m(A)=0$) implies that there is an $m$-measurable function $p(\lambda|Q)$ on $\Lambda$ such that for all $A\in \Omega$
\[ P_a(A|Q) = \int_A p(\lambda|Q) {\mathrm dm}(\lambda) \, . \]
In integration theory the function $p(\lambda|Q)$ is called the Radon-Nikodym derivative \cite{Rudin_RealAnalysis}.

We further suppose that the probability of obtaining result $r$ given measurement $M$ and ontic state $\lambda$ is given by an $m$-measurable function $p(r|M,\lambda)$. Then we have
\begin{equation}
P(r|M,Q) = \int_{\Lambda} p(r|M,\lambda) p(\lambda|Q) {\mathrm dm}(\lambda) \, .
\end{equation}
We can define the support $Sup_Q$ of the function $p(\lambda|Q)$ as the complement of the set on which it vanishes. Since $m(Sup_Q)=m(S_m\cap Sup_Q)$ there will be no loss of generality if we restrict the support to $S_m$. (Note that if the measure is not absolutely continuous, we cannot introduce in general the concept of ``support'' of the measures $P(-|Q)$ as it requires a topology on $\Lambda$ which we do not assume).

Given the above framework, we now state and prove the results presented in the main text in the case of  absolutely continuous measures. This includes the case where $\Lambda=\mathbb{R}^n$ with $m$ the standard measure on $\mathbb{R}^n$.

\begin{definition}[$\psi$-epistemic model] A statistical model is $\psi$-epistemic if there exist two preparation procedures $Q_1$ and $Q_2$ giving two distinct quantum states such that $m(Sup_{Q_1} \cap Sup_{Q_2}) > 0$. \end{definition}

\begin{definition}[Continuity] A $\psi$-epistemic model is $\delta$-continuous if, given the preparation Q of the pure quantum state $\psi_Q$, there is a set $A^\delta \in \Lambda$ (which can depend on $Q$) such that for all states $\phi$ satisfying $1-|{\langle{\phi}|{\psi}\rangle}| \leq \delta$, $A^\delta \subset Sup_{Q_\phi}$ and $\int_{A^\delta} p(\lambda|Q_\phi) {\mathrm dm}(\lambda) > 0$. Note that this implies that $m(A^\delta)>0$. A model is continuous, if it is $\delta$-continuous for some $\delta>0$. \end{definition}

Why is $\delta$-continuity a ``natural'' condition? Informally the reason is that --except in rare occasions-- physical quantities depend continuously on their parameters. Discontinuities are generally attributable to mathematical idealisation. To illustrate this in the present case, suppose that we can write $Q=(Q',\psi_Q)$, where $\psi_Q$ is the pure state produced by $Q$, and $Q'$ are additional variables specifying the preparation. Furthermore suppose that $\lambda, Q', \psi_Q$ all belong to compact subsets of real euclidean space $\mathbb{R}^n$, and that the probability distribution $p(\lambda|Q',\psi_Q)$ is a continuous function of all its parameters. Then it is uniformly continuous, and hence $\delta-$continuous for some $\delta>0$.

\begin{proof}[\textbf{Proof of Theorem 1}] Let the dimension of the quantum state space be $d$. We consider the same $d$ independent non-orthogonal states $\psi_k = \frac{1}{\sqrt{d-1}}\sum_{j\neq k} |j\rangle$ as in the main text. Suppose the model is $\delta$-continuous for $\delta \geq 1-\sqrt{(d-1)/d}$. Then there exists a set $A^\delta \subset Sup_{\psi_k}$ for all $k=1,\ldots,d$ with $m(A^\delta)>0$. Define a new function $\tilde{p}(\lambda) =\min_k p(\lambda|{\psi_k})$. Then we have
\begin{eqnarray} \label{eq:no-goGen}
\sum_k p(k|{\psi_k}) &=& \int_\Lambda \sum_k p(k|\lambda) p(\lambda|{\psi_k}) {\mathrm dm}(\lambda) \nonumber \\
&\geq& \int_{A^\delta} \sum_k p(k|\lambda) p(\lambda|{\psi_k}) {\mathrm dm}(\lambda) \nonumber \\
&\geq& \int_{A^\delta} \sum_k p(k|\lambda) \tilde{p}(\lambda) {\mathrm dm}(\lambda) \nonumber \\
&=& \int_{A^\delta} \tilde{p}(\lambda) {\mathrm dm}(\lambda) \, .
\end{eqnarray}
It remains to show that $\int_{A^\delta} \tilde{p}(\lambda) {\mathrm dm}(\lambda)$ is strictly positive. To this end we note that for each function $p(\lambda|{\psi_k})$ there must be a positive integer $N_k$ such that the set $T_{N_k} = \{\lambda\in A^\delta|p(\lambda|{\psi_k}) > 1/N_k\}$ satisfies $m(T_{N_k})>0$. Otherwise $\int_{A^\delta} p(\lambda|{\psi_k}){\mathrm dm}(\lambda)=0$, contrary to our hypothesis. Let $M=\max\{N_k\}$. Then for the set $T_M\equiv \{\lambda\in A^\delta|\tilde{p}(\lambda)>1/M\}$ we have $T_M\supset T_{N_K}$ for all $k$ and hence $m(T_M)>0$. We conclude that
\[ \int_{A^\delta} \tilde{p}(\lambda){\mathrm dm}(\lambda)\geq \frac{1}{M} m(T_M)>0 \, . \]
We thus arrive at the contradiction $0 = \sum_k p(k|{\psi_k})>0$. 
\end{proof}

\begin{definition}[Separability.] Consider a preparation $Q$ of a physical system and any subset of ontic states $A \in \Lambda$ that occurs with non-zero probability $P(A|Q)>0$. Consider $n$ identical copies of the preparation devices, which when used yield preparation $Q^n=(Q,\ldots,Q)$. Then the product set $A^n=A\times\ldots\times A$ occurs with non-zero probability: $P(A^n|Q^n)>0$. \end{definition}

Note that according to this definition the space of ontic states for $n$ independent preparations may be larger than the product $\Lambda\times\ldots\times\Lambda$. This does not affect the definition of separability.

\begin{proof}[\textbf{Proof of Theorem 2}]
The proof of Theorem 2 in the case of absolutely continuous measures is a simple combination of the arguments given in the main text in the case of discrete singular measures, and the proof of Theorem 1 for the case of absolutely continuous measures.
\end{proof}

\section{Relation between our continuity assumption and ontic indifference}\label{appB}

Hardy recently proposed a no-go theorem for $\psi$-epistemic models for a
single system that satisfies \emph{ontic indifference}, or \emph{restricted
ontic indifference} (see definitions below) \cite{hardy12}. Here we show that ontic
indifference is an extremely strong condition for $\psi$-epistemic
models, since it implies that there are ontic states $\lambda$ that are shared by all
quantum states in the Hilbert space. It is therefore immediately in contradiction
with the predictions of quantum theory since this implies that orthogonal
quantum states have overlapping ontic distributions. Restricted ontic indifference does
not suffer from this problem, but on the other hand it breaks the
unitary invariance of Hilbert space by promoting one state to a specific status.
We show how the proof methods we use to rule out $\delta$-continuous
$\psi$-epistemic models can be easily transposed to provide no-go results
for $\psi$-epistemic models obeying restricted ontic indifference.

\begin{definition}[Ontic indifference \cite{hardy12}] Any quantum transformation
on a system which leaves unchanged any given pure state, $\vert\psi\rangle$,
can be performed in such a way that it does not affect the underlying
ontic states $\lambda\in\Lambda_{\vert\psi\rangle}$ in the ontic
support of that pure state. By the ontic support of a given state,
$\vert\psi\rangle$, one means the set $\Lambda_{\vert\psi\rangle}$
of ontic states $\lambda$ which might be prepared when the given
pure state is prepared.\end{definition}

\begin{observation} Ontic indifference implies the following fact: if two
distinct states $\vert\psi\rangle\neq\vert\psi'\rangle$ share the
same ontic state $\lambda$, that is $\lambda\in\Lambda_{\vert\psi\rangle}$
and $\lambda\in\Lambda_{\vert\psi'\rangle}$, and if $U$ is a unitary
transformation that leaves $\vert\psi\rangle$ invariant, $U\vert\psi\rangle=\vert\psi\rangle$,
then $\lambda\in\Lambda_{U\vert\psi'\rangle}$ belongs to the ontic
support of $U\vert\psi'\rangle$. \end{observation}

This follows from considering the following 4 preparations: a specific
preparation $P_{1}$ of $\vert\psi\rangle$, the preparation $P_{1}$
followed by unitary $U$ (that leaves $\vert\psi\rangle$ invariant), a specific preparation $P_{2}$ of $\vert\psi'\rangle$, and the preparation
$P_{2}$ followed by unitary $U$ (which prepares $U\vert\psi'\rangle$).
Then we note that preparation $P_{1}$ sometimes yields the ontic
state $\lambda$, that $U$ leaves $\lambda$ invariant, that preparation
$P_{2}$ sometimes yields the ontic state $\lambda$, hence preparation
$P_{2}$ followed by $U$ sometimes yields the ontic state $\lambda$,
and therefore $\lambda\in\Lambda_{U\vert\psi'\rangle}$. 

This observation has the following consequence, which implies directly that no $\psi$-epistemic model can reproduce the predictions of quantum theory (since it implies in particular that orthogonal quantum states share a common real state $\lambda$). 
\begin{theorem} If any two distinct states $\vert\psi_{1}\rangle\neq\vert\psi_{2}\rangle$
share the same ontic state $\lambda$, that is $\lambda\in\Lambda_{\vert\psi_{1}\rangle}$
and $\lambda\in\Lambda_{\vert\psi_{2}\rangle}$, and if ontic indifference
holds, then all states in the Hilbert space share the ontic state
$\lambda$, that is $\lambda\in\Lambda_{\vert\psi\rangle}$ for all
$\vert\psi\rangle$. \end{theorem}
\begin{proof}
To prove this consider the following sets $S_{k}$ ($k\geq 0$) of states, where $S_{0}=\left\{ \vert\psi_{1}\rangle,\vert\psi_{2}\rangle\right\}$ and 
$S_{k}=\{|\psi\rangle:$ there exist $|\phi\rangle,|\phi'\rangle\in S_{k-1}$ and a unitary $U$ such that $U|\phi\rangle=|\phi\rangle$, $U|\phi'\rangle=|\psi\rangle\}$. That is, $S_{k}$ is the set of states that can be obtained from $S_{k-1}$ by a unitary that leaves one of the states in $S_{k-1}$ invariant.

By hypothesis, both states in $S_{0}$ share the same ontic state $\lambda$.
Hence by the above observation, all states in $S_{k},$ for all $k=0,1,2,...$,
also share the state $\lambda$.
Now consider the case of a 2-dimensional Hilbert space. If $|\langle\psi_{1}\vert\psi_{2}\rangle\vert=\cos\theta$
are at an angle $\theta$ from each other (as measured on the Bloch
sphere), then it is easy to show (just draw the successive sets $S_{k})$
that the size of $S_{k}$ increases exponentially, such that when
$k\geq O(\log\theta)$, $S_{k}$ comprises the whole Bloch sphere.
The case of a 2-dimensional Hilbert space helps to visualize what
is happening, but is not a restriction, since the linear span of the
set $S_{1}$ is the whole Hilbert space. 
\end{proof}

We now show how a simple application of our results provides a no-go theorem for $\psi$-epistemic models satisfying restricted ontic indifference.

\begin{definition}[Restricted ontic indifference \cite{hardy12}] Any quantum transformation
on a system which leaves unchanged the pure state $\vert\psi\rangle$
can be performed in such a way that it does not affect the underlying
ontic states, $\lambda\in\Lambda_{\vert\psi\rangle}$, in the ontic
support of that pure state. \end{definition}

\begin{theorem}
The quantum predictions are incompatible with $\psi$-epistemic models satisfying the assumptions of restricted ontic indifference and separability, and such that there exists a state $|\psi'\rangle\neq|\psi\rangle$ sharing a common ontic state $\lambda$ with $|\psi\rangle$, where $|\psi\rangle$ is the particular state in the definition of restricted ontic indifference.
\end{theorem}

\begin{proof}
Let $d\geq 2$ be an integer such that $|\langle\psi'|\psi\rangle|\leq\sqrt{\frac{d-1}{d}}$. 
Add an ancilla $|0\rangle$ to the space such that the dimension of the full Hilbert
space is greater or equal to $d$. By the assumption of separability the states $|\Psi\rangle=|\psi\rangle|0\rangle$ and  $|\Psi'\rangle=|\psi'\rangle|0\rangle$ share a common ontic state $\lambda$. 
Choose a basis such that $\vert\Psi\rangle=\frac{1}{\sqrt{d}}\sum_{i=1}^{d}\vert i\rangle$,
and coefficients $\alpha_{i}\in\mathbb{R}$, $i=1,\ldots,d$ such
that $\alpha_{1}=0$, $\sum_{i=2}^{d}\alpha_{i}^{2}=1$, $\frac{1}{\sqrt{d}}\sum_{i=2}^{d}\alpha_{i}=|\langle\Psi'|\Psi\rangle|$
(this is always possible since $|\langle\Psi'|\Psi\rangle|\leq\sqrt{\frac{d-1}{d}}$).
Define the states $\vert\Phi_{k}\rangle=\sum_{i=1}^{d}\alpha_{i-k}\vert i\rangle$.
(Note that these states coincide with the states $\Psi_{k}$ used
in the main text if $\alpha_{i}=\frac{1}{\sqrt{d-1}}$, $i=2,\ldots,d$).
The Observation above implies that all the states $|\Phi_{k}\rangle$ share
the ontic state $\lambda$, since $|\langle\Phi_{k}|\Psi\rangle|=|\langle\Psi'|\Psi\rangle|$. We then obtain a contradiction with the predictions of quantum theory
exactly as in the proof of Theorem~1 given in the main text.
\end{proof}

\end{document}